\documentclass[11pt]{article} 
\usepackage[text={6.5in,9in}]{geometry}
\usepackage{amsmath,latexsym,amssymb,color,amsthm}
\usepackage{ifthen,graphics,epsfig}
\usepackage{color}
\usepackage{xspace}
\usepackage{enumitem}
\usepackage{url}
\usepackage{amssymb}
\usepackage{pifont}
\usepackage{authblk}
\usepackage{booktabs}
\usepackage{wrapfig}
\usepackage[colorlinks=true,citecolor=gray]{hyperref}
\usepackage{caption}
\usepackage{subcaption}
\usepackage[table]{xcolor}

\usepackage{float}
\newfloat{algorithm}{thp}{lop}
\floatname{algorithm}{Algorithm}
\newcommand{\Xomit}[1]{}

\usepackage{marginnote}
\usepackage[table]{xcolor}
\usepackage{todonotes}
\usepackage{menukeys}

\newcommand{\remove}[1]{}

\newlength {\squarewidth}

\newtheorem{theorem}{Theorem}
\newtheorem{lemma}{Lemma}

\newcommand{\toto}{xxx}

\newcounter{linecounter}
\newcommand{\linenumbering}{\ifthenelse{\value{linecounter}<10}{(0\arabic{linecounter})}{(\arabic{linecounter})}}
\renewcommand{\line}[1]{\refstepcounter{linecounter}\label{#1}\linenumbering}
\newcommand{\resetline}[1]{\setcounter{linecounter}{0}#1}
\renewcommand{\thelinecounter}{\ifnum \value{linecounter} > 9\else 0\fi \arabic{linecounter}}

\newcommand{\signed}[1]{\theta_{#1}}
\newcommand{\signedm}[2]{\big( #2, \theta_{#1} \big)}
\newcommand{\aux}[2]{\textsc{aux}[#1](#2)}
\newcommand{\coin}[1]{\textsc{coin}[#1]}
\newcommand{\coinval}{\textsc{c\_val}}
\newcommand{\coinmap}{coin\_map}
\newcommand{\auxproof}[1]{\langle#1\rangle}
\newcommand{\proofs}{\mathit{proofs}}
\newcommand{\isvalid}{{\sf is\_valid}}
\newcommand{\auxvalues}{\mathit{aux\_values}}

\newcommand{\send}{\mathit{\sf send}}
\newcommand{\sto}{\mathit{\sf to}}
\newcommand{\receive}{\mathit{\sf receive}}
\newcommand{\broadcast}{\mathit{\sf broadcast}}

\newcommand{\ttrue}{\mathit{\tt true}}
\newcommand{\tfalse}{\mathit{\tt false}}

\newcommand{\binpropose}{{\sf bin\_propose}}

\newcommand{\wait}{{\sf wait\_until}}
\renewcommand{\return}{{\sf return}}

\newcommand{\BAMP}{{\cal BAMP}_{n,t}}

\newcommand{\decide}{{\sf decide}}

\title{A Simple and Efficient Asynchronous Randomized Binary Byzantine Consensus Algorithm}

\author[1]{Tyler Crain}
\affil[1]{tcrainwork@gmail.com}

\begin{document}

\maketitle

\begin{abstract}
  This paper describes a simple and efficient asynchronous Binary Byzantine faulty tolerant
  consensus algorithm.
  In the algorithm, non-faulty nodes perform an initial broadcast followed by a executing a series of rounds
  each consisting of a single message broadcast plus the computation of a global random coin using
  threshold signatures.
  Each message is accompanied by a cryptographic proof of its validity.
  Up to one third of the nodes can be faulty and termination
  is expected in a constant number of rounds.
  An optimization is described allowing the round message plus the coin message to be combined,
  reducing rounds to a single message delay.
  Geo-distributed experiments are run on replicates in ten data center regions
  showing average latencies as low as $400$ milliseconds.
\end{abstract}













\section{Introduction and related work.}
Binary byzantine consensus concerns the problem of getting a set of distinct processes
distributed across a network to agree on a single binary value $0$ or $1$ where processes
can fail in arbitrary ways.
It is well known that this problem is impossible in an asynchronous network with at least
one faulty process~\cite{FLP85}. To get around this, algorithms can employ
randomization~\cite{A03, BO83, B87, BT83, CR93, FP90, KS16, MMR14, MMR15, PCR14, R83, T84},
or rely on an additional synchrony assumption~\cite{DDS87, DLS88}.
Randomized algorithms normally rely on the existence of a local or global random coin.
The output of local coin is only visible to an individual process, while the output of a global
coin is visible to all processes, but only once a threshold of processes have participated in computing the coin.

The algorithm presented in this paper uses a global coin, assumes at most one third of the processes
are faulty (a well know lower bound~\cite{LSP82}), and terminates in expected $O(1)$ number of message delays.
While there are many algorithms that solve this problem with the same guarantees~\cite{A03,CKS05,MR17},
this algorithm focuses on simplicity and efficiency.
Namely, it starts with each process broadcasting an initial proposal, then executing a series
of rounds that consist of two all to all message broadcasts.
The first being to distribute processes current binary estimates, and the second being used to compute
the output of the global coin.

The design of the algorithm is primarily based on two previous algorithms;~\cite{CKS05} and~\cite{MMR14}.
While these algorithms provide similar theoretical guarantees, they are slightly more complex/costly.
In this paper, like in~\cite{CKS05}, threshold signatures~\cite{CH89, D88, DF90, R98} are used to
implement the global coin, and a set of cryptographic signatures are included with each message proving
its validity.
The algorithm presented here differs in that it requires one less all to all message broadcast per round.
Similar to the randomized algorithm of~\cite{MMR14} this work relies on a global coin for correctness.
Differently,~\cite{MMR14} does not include cryptographic signatures with each message,
but requires up to $2$ additional message broadcasts per round and furthermore is not fully asynchronous
as it requires a fair scheduler to ensure termination in all cases~\cite{MMR15}.

While the binary consensus problem only allows process to agree on a single binary value,
there exist many reductions to multi-value consensus~\cite{MR17, MRT00, TC84, ZC09} allowing processes to agree on arbitrary values.
Furthermore many algorithms~\cite{BSA14, CL02} exists that solve multi-value consensus directly
through the use of types of synchrony assumptions to ensure termination.
Additionally, algorithms exists that make many different assumptions about the model
such as synchrony~\cite{FM97}, different fault models~\cite{LVCQV16, MA06, PSL80},
solve different definitions of consensus~\cite{NCV05}, and so on.

It should should be noted that the design of this algorithm follows closely the algorithm of~\cite{C20},
but~\cite{C20} relies on partial synchrony~\cite{DLS88} for termination through the use of
a weak round coordinator and timeout.
In most cases that algorithm terminates much faster and does not require threshold signatures.
It is therefore suggested to use that algorithm over this one, unless a truly
adversarial asynchronous network is expected.

\section{A Byzantine Computation Model.}
\label{sec:model}

This section describes the assumed computation model.
For simplicity we assume idealized cryptographic assumptions.

\paragraph{Asynchronous processes.}
The system is made up of a set $\Pi$ of $n$ asynchronous sequential processes,
namely $\Pi = \{p_1,\ldots,p_n\}$; $i$ is called the ``index'' of $p_i$. 
``Asynchronous'' means that each process proceeds at its own speed,
which can vary with time and remains unknown to the other processes.
``Sequential'' means that a process executes one step at a time.
This does not prevent it from executing several threads with an appropriate
multiplexing. 
%
Both notations
$i\in Y$ and $p_i\in Y$ are used to say that $p_i$ belongs to the set $Y$.

\paragraph{Communication network.}
\label{sec:basic-comm-operations}
The processes communicate by exchanging messages through
an asynchronous reliable point-to-point network. ``Asynchronous''  means that
there is no bound on message transfer delays, but these delays are finite.
``Reliable'' means that the network does not lose, duplicate, modify, or
create messages. ``Point-to-point'' means that any pair of processes
is connected by a bidirectional channel.
%
A process $p_i$ sends a message to a process $p_j$ by invoking the primitive 
``$\send$ {\sc tag}$(m)$ $\sto~p_j$'', where {\sc tag} is the type
of the message and $m$ its content. To simplify the presentation, it is
assumed that a process can send messages to itself. A process $p_i$ receives 
a message by executing the primitive ``$\receive()$''.
The macro-operation $\broadcast$ {\sc tag}$(m)$ is  used as a shortcut for
``{\bf for each} $p_i \in \Pi$  {\bf do} $\send$ {\sc tag}$(m)$ $\sto~p_j$
{\bf end for}''. 

\paragraph{Signatures.}
Asymmetric cryptography allow processes to sign messages.
Each process $p_i$ has a public key known by everyone and a private key
known only by $p_i$.
All messages are signed using the private key and can be validated by any process
with the corresponding public key, allowing the process to identify the signer of the message.
Signatures are assumed to be unforgeable.
A process will ignore any message that is malformed or contains an invalid signature.

\paragraph{$(n-t)$ non-interactive threshold signatures.}
Given the set of $n$ processes,
taking $n-t$ signatures of the same message from $n-t$ different processes can be combined to
generate a \emph{unique} threshold signature that can be verified by a threshold public key
known by everyone.
Any set of $n-t$ signatures of the same message from $n-t$ different processes
generates the same threshold signature.
Threshold signatures are assumed to be unforgeable and no set of less than $n-t$ nodes
can generate them.

\paragraph{Random oracle.}
A random oracle~\cite{FS86} is assumed giving us access to an ideal cryptographic hash function.
This function takes a set of bits as input and deterministically outputs a random
fixed length set of bits.

\paragraph{Failure model.}
Up to $t$ processes can exhibit a {\it Byzantine} behavior~\cite{PSL80}.
 A Byzantine process is a process that behaves
arbitrarily: it can crash, fail to send or receive messages, send
arbitrary messages, start in an arbitrary state, perform arbitrary state
transitions, etc. Moreover, Byzantine processes can collude 
to ``pollute'' the computation (e.g., by sending  messages with the same 
content, while they should send messages with distinct content if 
they were non-faulty). 
A process that exhibits a Byzantine behavior is called {\it faulty}.
Otherwise, it is {\it non-faulty}.  
Let us notice that, 
as unforgeable signatures are used
no Byzantine process can impersonate another process.   
Byzantine processes can control the network by modifying
the order in which messages are received, but they cannot
postpone forever message receptions.  


\section{Binary Byzantine Consensus.}
\label{sec:byz-consensus}

\subsection{The Binary Consensus Problem.}


%

In the binary consensus problem processes input a value to the algorithm, called their \emph{proposal},
run an algorithm consisting of several rounds,
and eventually output a binary value called their \emph{decision}.
Let $\cal V$ be the set of values that can be proposed.
While  $\cal V$ can contain any number ($\geq 2$) of values
in multi-valued consensus, it contains only two values in binary consensus, 
e.g., ${\cal V} =\{0,1\}$.
Assuming that
each non-faulty process proposes a value, the binary Byzantine consensus (BBC) problem is for 
each of them to
decide on a value in such a way that the following properties are
satisfied:
\begin{itemize}
\item BBC-Termination. Every non-faulty process eventually decides on a value.
\item BBC-Agreement.   No two non-faulty processes decide on different values.
\item BBC-Validity.  If all non-faulty processes propose the same value, no
other value can be decided.
\end{itemize}

\paragraph{Notations.}

\begin{itemize}
\item The acronym ${\BAMP}[\emptyset]$ is used to denote the 
  basic Byzantine Asynchronous Message-Passing computation model;
  $\emptyset$ means that there is no additional assumption. 
\item The basic computation model strengthened with the additional constraint $t<n/3$
  is denoted ${\BAMP}[t<n/3]$.
\item A signature of process $i$ is $\signed{i}$.
\item A message $m$ signed by process $i$ is $\signedm{i}{m}$.
\end{itemize}


\subsection{A Safe and Live Consensus Algorithm in ${\BAMP}[t<n/3]$.}
\label{ssec:live-bbc}

\paragraph{Message types.}
The following message types are used by the consensus.
\begin{itemize}
\item $\aux{r}{v}$. An {\sc aux} message contains a round number $r$ and a binary value $v$.
\item $\auxproof{\signedm{i}{\aux{r}{v}},\proofs}$. A tuple containing an {\sc aux} message signed
  by process $i$ and a set $\proofs$ containing signed
  {\sc aux} messages from a previous round that are used to prove $v$ is a \emph{valid} binary proposal for round $r$.
\item $\signedm{i}{\coin{r}}$.  A message for round $r$ signed by process $i$ that will be used to generate
  random global coin outputs.
\end{itemize}

\paragraph{Valid Notation.}
For a given round $r \geq 1$ a binary value $b$ is \emph{valid} if
$b$ has been proposed by a non-faulty process and
$\neg b$ has not been decided in any round before $r$.
An $\auxproof{\signedm{i}{\aux{r}{v}},\proofs}$ is valid if binary value $v$ is valid in round $r$.
The algorithm describes a function that is used to compute the validity of a message
given $r$, $v$, and $\proofs$ as input.

\paragraph{An $(n-t)$-threshold random global coin.}
The existence of a random global coin is assumed for both correctness and termination of the algorithm.
The coin is ``flipped'' when processes participate in computing the output of the coin for a given round of the algorithm.
The following properties are ensured by the coin.
\begin{itemize}
\item {\bf c-binary}. The output of a coin flip is a binary value.
\item {\bf c-threshold}. The output of a coin flip is not revealed until at least $n-t$ processes participate in the coin flip.
\item {\bf c-global}. All processes observe the same output of a coin flip.
\item {\bf c-random}. The output of the coin flip is random meaning that before $n-t$ processes have participated in flipping the coin
  then no process can correctly guess the output of the coin with probability greater than $1/2$.
\item {\bf c-flip}. The coin can be flipped any number of times.
\end{itemize}

In the algorithm a coin is flipped every round as follows:
When a process signs and broadcasts a $\signedm{i}{\coin{r}}$ message, the process
is considered to have participated in flipping the coin for round $r$.
The output of the coin is generated by taking the first bit of the cryptographic hash computed using the $(n-t)$ threshold signature
of the $\signedm{i}{\coin{r}}$ message as input.
Given that a threshold signatures are unique, cannot be computed with less than $(n-t)$ signatures, and that
the output of the cryptographic hash is random, the properties of the coin are ensured.
The algorithm~\cite{CKS05} generates random coin values in a similar manner.

\paragraph{Variables.}
The following variables are used throughout all rounds of the consensus.
\begin{itemize}
\item $r_i$. Current round number of process $i$.
\item $est_i$. Current estimate at process $i$. It can either be a binary value ($0$ or $1$)
  or it can be the special value $\coinval$ meaning the estimate will chosen as the result of the coin flip of round $r_i$.
\item $\coinmap_i$. Map from round to the binary value corresponding to the result of the coin flip for that round at process $i$.
\item $\auxvalues_i$. Set of valid signed {\sc aux} messages received by process $i$
  throughout all rounds of the consensus.
\end{itemize}

\begin{figure*}[ht!]
\centering{
\fbox{
\begin{minipage}[t]{150mm}
\footnotesize
\renewcommand{\baselinestretch}{2.5}
\resetline
\begin{tabbing}
aaaA\=aaA\=aaaA\=aaaaaaaaaA\kill

{\bf opera}\={\bf tion} ${\binpropose}(v_i)$ {\bf is}\\

\line{BBC-01} \> $r_i \leftarrow 0$; $\auxvalues_i \leftarrow \emptyset$; $\coinmap_i \leftarrow \emptyset$;\\


\line{BBC-02} \> ${\broadcast}$ $\auxproof{\signedm{i}{\aux{r_i}{v_i}},\emptyset}$;
{\it \scriptsize  \hfill \color{gray}{// Broadcast the initial proposal}}\\

\line{BBC-03} \> $\wait$ $(n-t)$ valid $\aux{r_i}{}$ messages have been received from $(n-t)$ different processes; \\

\line{BBC-04} \> {\bf if} \= (at least $(t+1)$ valid $\aux{r_i}{0}$ messages have been received from $(t+1)$ different processes) \\

\line{BBC-05} \>\> {\bf then} $est_i \leftarrow 0$ \\

\line{BBC-06} \>\> {\bf else} $est_i \leftarrow 1$ \\

\line{BBC-07} \> {\bf end if} \\

\line{BBC-08} \> {\bf while} $(\ttrue)$  {\bf do} \\

\line{BBC-09} \>\> $r_i \leftarrow r_i+1$;\\

\line{BBC-10} \>\> $proofs_i \leftarrow $ compute $proofs_i$ as a set of signed {\sc aux} messages from $\auxvalues_i$ that satisfy the \\

\>\>   $~~~~~~~~~$  $\isvalid$ predicate for binary value $est_i$ and round $r_i$; \\

\line{BBC-11} \>\> ${\broadcast}$ $\auxproof{\signedm{i}{\aux{r_i}{est_i}},proofs}$; \\

\line{BBC-12} \>\> $\wait$ $(n-t)$ valid $\aux{r_i}{}$ messages have been received from $(n-t)$ different processes; \\

\line{BBC-13} \>\> {\bf if} \= ($\exists$ $b\_val \in \{0,1\}$ such that $(n-t)$ valid $\aux{r_i}{b\_val}$ messages have been  \\
\>\>   $~~~~~~~~~~~$ received from $(n-t)$ different processes) \\

\line{BBC-14} \>\>\> {\bf then} $est_i \leftarrow b\_val$  \\

\line{BBC-15} \>\>\> {\bf else} $est_i \leftarrow \coinval$
{\it \scriptsize  \hfill \color{gray}{// $est_i$ will take the value of the coin when it is revealed}} \\

\line{BBC-16} \>\> {\bf end if} \\





\line{BBC-17} \>\> ${\broadcast}$ $\signedm{i}{\coin{r_i}}$; \\

\line{BBC-18} \>\> $\wait$ $(n-t)$ valid $\coin{r_i}$ messages have been received from $(n-t)$ different processes; \\

\line{BBC-19} \>\> $\coinmap_i[r_i] \leftarrow$ compute the first bit of the cryptographic hash of the $(n-t)$ threshold \\
\>\>   $~~~~~~~~~~~~~~~~$ signature of $\coin{r_i}$; \\

\line{BBC-20} \>\> {\bf if} \= ($n-t$) valid $\aux{r_i}{\coinmap_i[r_i]}$ messages have been received from $(n-t)$ different processes \\

\line{BBC-21} \>\>\> {\bf then} $\decide(\coinmap_i[r_i])$ if not yet done {\bf  end if} \\

\line{BBC-22} \>\> {\bf if} \= $est_i = \coinval$ {\bf then} $est_i \leftarrow \coinmap_i[r_i]$ {\bf end if} \\






\line{BBC-23} \> {\bf end while}; \\~\\

{\bf when} $\auxproof{\signedm{j}{\aux{r_j}{est_j}},proofs}$  {\bf is received}\\

\line{BBC-24} \>   
   {\bf if} \=
   $\big(\isvalid(r_j, est_j, proofs)\big)$   {\bf then} \\

   \line{BBC-25} \> \> $\proofs \leftarrow proofs \setminus $ $\{$any messages in $\proofs$ not needed to satisfy the $\isvalid$ predicate$\}$. \\
   
\line{BBC-26} \> \>  
  $\auxvalues_i \leftarrow \auxvalues_i \cup \{\signedm{j}{\aux{r_j}{est_j}}\} \cup proofs$; \\


  \line{BBC-27} \>  {\bf end if}.

\end{tabbing}
\normalsize
\end{minipage}
}
\caption{A safe algorithm for the binary Byzantine consensus in ${\BAMP}[t<n/3]$.}
\label{algo-BBC} 
\vspace{-1em}
}
\end{figure*}

\begin{figure*}[ht!]
\centering{
\fbox{
\begin{minipage}[t]{150mm}
\footnotesize
\renewcommand{\baselinestretch}{2.5}
\resetline
\begin{tabbing}
aaaA\=aaA\=aaaA\=aaaaaaaaaA\kill

{\bf pred}\={\bf icate} ${\isvalid}(r, est, proofs)$ {\bf is}\\

\line{IV-01} \> {\bf if} $(r=0)$   {\bf then} $\return(\ttrue)$ {\bf  end if};\\

\line{IV-02} \> {\bf if} $(r=1)$ {\bf then} \\

\line{IV-03} \>\> {\bf if} ($\exists$ signed messages $\aux{0}{est}$ from $t+1$ different processes in $\proofs$) \\

\>\>\>  {\bf then} $\return(\ttrue)$ \\

\>\>\> {\bf else} $\return(\tfalse)$ \\

\>\> {\bf end if} \\

\line{IV-04} \> {\bf end if}; \\

\line{IV-05} \> $prev\_r \leftarrow$ compute $prev\_r$ as the largest round smaller than $r$ where $\coinmap_i[prev\_r] = \neg est$ \\
\>\>   $~~~~~~~~$  or $0$ if no such round exists; \\


\line{IV-06} \> {\bf if} $(prev\_r = 0$ $\wedge$ $\exists$ signed messages $\aux{0}{est}$ from $t+1$ different processes in $proofs)$ \\

\>\>  {\bf then} $\return(\ttrue)$ \\





\line{IV-07} \> {\bf else if} $(\exists$ signed messages $\aux{prev\_r}{est}$ from $n-t$ different processes in $proofs)$ \\

\>\> {\bf then} $\return(\ttrue)$ \\

\line{IV-08} \> {\bf end if} \\

\line{IV-09} \> $\return(\tfalse)$.

\end{tabbing}
\normalsize
\end{minipage}
}
\caption{Algorithm for the $\isvalid$ predicate.}
\label{algo-is-safe} 
\vspace{-1em}
}
\end{figure*}

\paragraph{Algorithm description.}
Figures~\ref{algo-BBC} and~\ref{algo-is-safe} describe the pseudo-code
for the algorithm.
The $\binpropose$ operation of Figure~\ref{algo-BBC} contains the main loop of the algorithm.
The lines~\ref{BBC-24}-\ref{BBC-27} handle the reception of signed {\sc aux} messages.
The $\isvalid$ predicate of Figure~\ref{algo-is-safe} describes the procedure used
to check if a binary value is valid for a given round and a set $\proofs$ of signed {\sc aux} messages.

To start the consensus, each process $p_i$ calls $\binpropose$ with its initial binary proposal $v_i$ (Figure~\ref{algo-BBC}).
Line~\ref{BBC-01} initializes local variables, then on Line~\ref{BBC-02} the process
broadcasts a signed {\sc aux} message with round $0$, binary value $v_i$, and an
empty set for $\proofs$ as any round $0$ message is considered to be valid.
The process then waits until $n-t$ round $0$ {\sc aux} messages are received (line~\ref{BBC-03}).
An initial estimate is then chosen by taking a binary value that has at least $t+1$ broadcasters
(lines~\ref{BBC-04}-\ref{BBC-07}). This ensures that the estimate was broadcast by a non-faulty process.
The process then repeats the while loop of Lines~\ref{BBC-08}-\ref{BBC-23} for each round.

A round starts by incrementing the round counter on line~\ref{BBC-04}.
The process then uses the $\isvalid$ predicate to compute a set of signed {\sc aux} messages
that prove $est_i$ is valid in the current round (line~\ref{BBC-10}).
Next the process signs an {\sc aux} message for round $r$ with binary value $est_i$ and broadcasts
it along with the proofs generated on the previous line.
The process then waits until $n-t$ valid {\sc aux} messages are received from different processes
for the current round (line~\ref{BBC-12}).
The estimate of the valid binary values is then update on lines~\ref{BBC-13}-\ref{BBC-16} as follows:
First, if all of the received {\sc aux} messages contain the same binary value then the processes sets its estimate to this value,
otherwise it sets its estimate to $\coinval$, meaning that once the output of the coin flip for the current
round is revealed, the process will set its estimate to this value.
Given that $t < n/3$, receiving $n-t$ signed {\sc aux} messages of the same value ensures that any set of
$n-t$ signed {\sc aux} messages for that round will contain at least one {\sc aux} message supporting the same value,
thus all non-faulty processes will set their estimate to either this value or the value of the coin.


The process then participates in computing the value of the coin for this round by
broadcasting a {\sc coin} message (line~\ref{BBC-17}).
Next it waits until $n-t$ signed coin messages have been received, computes
the $n-t$ threshold signature of the message, inputs this value to the cryptographic hash function,
and takes first bit output as the value of the coin (lines~\ref{BBC-18}-\ref{BBC-19}).
Following this, if $n-t$ signed {\sc aux} messages have been received from different processes
with the same value as the coin, this value is decide (lines~\ref{BBC-20}-\ref{BBC-21}).
Finally, if the estimate of the process has been set to $\coinval$, it is updated to the value of the coin
(line~\ref{BBC-22}) and process continues to the next round.

Lines~\ref{BBC-24}-\ref{BBC-27} describe what happens when a signed {\sc aux} message and its proofs are received.
If the $\isvalid$ predicate indicates that this message is valid, then the signed {\sc aux} message and its
proofs are added to the $\auxvalues_i$ set (lines~\ref{BBC-22}-\ref{BBC-24}).
Line~\ref{BBC-24} ensures that no invalid messages are added to $\auxvalues_i$.

\paragraph{Is\_valid predicate description.}
Figure~\ref{algo-is-safe} describes the $\isvalid$ predicate that is called by Algorithm~\ref{algo-BBC}
to check if a binary value is valid.  It takes as input a round $r$, a binary value $est$
and a set of signed {\sc aux} messages in $\proofs$.
As previously mentioned, the predicate
should return $\ttrue$ if $\proofs$ ensures that
(i) $est$ was proposed by a non-faulty process
and (ii) $\neg est$ has not been
decided by any non-faulty process in any round before $r$.
Otherwise $\tfalse$ should be returned.

For round $0$ the predicate immediately returns $\ttrue$ as any initial proposal is valid (line~\ref{IV-01}).
For round $1$, as no value can be decided in round $0$, the predicate returns $\ttrue$
if $\proofs$ contains at least $t+1$ round $0$ messages with binary value $est$ (line~\ref{IV-03}), i.e.
if (i) is satisfied.

For any other round $r > 1$, the process computes the largest round $prev\_r$ such that $prev\_r < r$ and the value of the coin
in $prev\_r$ was $\neg est$, otherwise $prev\_r$ is set to $0$ if no such round exists (line~\ref{BBC-05}).
Given line~\ref{BBC-21} of Figure~\ref{algo-BBC} we know that a value can only be decided if $n-t$ messages
are received matching the value of the coin.
Thus if $prev\_r$ is $0$, we know $\neg est$ could not have been decided before round $r$ and the predicate is satisfied
as long as $proofs$ contains $t+1$ signed $\aux{0}{est}$ messages.
Otherwise if $prev\_r > 0$ the predicate is satisfied if $proofs$ contains $n-t$ signed $\aux{prev\_r}{est}$ messages.
In this case, as $t < n/3$ no process could have received $n-t$ signed $\aux{prev\_r}{\neg est}$ messages and therefore
no non-faulty process could have decided
$\neg est$ in rounds from $prev\_r$ until $r$.
An argument by induction can be then made that $\neg est$ was not decided in any previous round.

If none of these cases are met then $\tfalse$ is returned.

\subsection{Proofs.}

This section shows that the algorithm presented in Figure~\ref{algo-BBC} solves the Binary consensus problem in ${\BAMP}[t<n/3]$
through a series of lemmas.

\begin{lemma}
  \label{lem:maj}
  For a given round $r$ there can be at most one binary value $b$ for which there exists at least $n-t$
  signed $\aux{r}{b}$ messages from different processes.
\end{lemma}
\begin{proof}
  This follows from the fact that there are at most $t < n/3$ faulty processes and that non-faulty processes
  sign and broadcast at most one {\sc aux} message per round.
\end{proof}

The following lemma shows that processes will receive enough messages in every round to progress to the following round.

\begin{lemma}
  \label{lem:enough}
  At any non-faulty process $p_i$ with estimate $est_i$ and round $r_i > 0$, $p_i$ will (eventually) receive enough valid messages to
  satisfy the $\isvalid$ predicate of Figure~\ref{algo-is-safe} for the $r_i$ and $est_i$.
\end{lemma}
\begin{proof}
  
  By line~\ref{IV-01} of the $\isvalid$ predicate all signed round $0$ {\sc aux} messages are valid and
  by line~\ref{BBC-02} of Figure~\ref{algo-BBC} all non-faulty processes sign and broadcast a round $0$ {\sc aux} message.
  All non-faulty processes will then receive at least $n-t$ signed round $0$ {\sc aux} messages from different processes.
  Given that $t < n/3$, of these $n-t$ messages, at least $t+1$ messages supporting a single binary
  value will be received and the process will set its estimate to this value on lines~\ref{BBC-04}-\ref{BBC-06},
  satisfying line~\ref{IV-03} of the $\isvalid$ predicate for round $1$.
  All non-faulty processes will then sign and broadcast a valid {\sc aux} message for round $1$,
  participate in computing the coin, and advance to round $2$.

  Let the output of the coin for round $1$ be some binary value $b\_val_1$.
  In round $2$ non-faulty processes will receive at least $n-t$ signed valid round $1$ {\sc aux} messages from different processes.
  If $n-t$ of these messages are of the form $\aux{1}{\neg b\_val_1}$, then by line~\ref{IV-07} of Figure~\ref{algo-is-safe} the $\isvalid$ predicate
  is satisfied for round $2$. Additionally the process will set its estimate to $\neg b\_val_1$ on line~\ref{BBC-14} of Figure~\ref{algo-BBC}.
  Otherwise, at least one of the valid signed {\sc aux} messages received must be of the form $\aux{1}{b\_val}$
  and the estimate is set to $b\_val_1$ (the value of the coin) on line~\ref{BBC-15}.
  By line~\ref{BBC-24} of Figure~\ref{algo-BBC}
  this message must contain proofs generated by the $\isvalid$ predicate supporting binary value $b\_val_1$ for round $1$.
  This can only happen on line~\ref{IV-03} of Figure~\ref{algo-is-safe} by including $t+1$ messages of the form $\aux{0}{b\_val}$.
  Notice then, that given the value for the coin for round $1$ is $b\_val_1$, $prev_r$ will be computed as $0$
  on line~\ref{IV-05}, and by line~\ref{IV-07} these proofs also satisfy the $\isvalid$ predicate for round $2$.
  Thus, all non-faulty processes will then sign and broadcast a valid {\sc aux} message for round $2$ and advance to round $3$. 

  Now assume by induction all non-faulty processes have received enough valid messages to satisfy the $\isvalid$ predicate
  for a round $r-1$. All processes will then sign and broadcast a valid {\sc aux} message on line~\ref{BBC-11} of
  Figure~\ref{algo-BBC}, participate in computing the coin, and advance to round $r$.
  As a result all non-faulty processes will receive at least $n-t$ valid signed {\sc aux} messages from round $r-1$.

  Let the output of the coin for round $r$ be some binary value $b\_val_r$.
  
  First consider a non-faulty process whose estimate $est_i$ in round $r$ was set to the value $\neg b\_val_r$ in round $r-1$.
  In this case, by line~\ref{BBC-12} of Figure~\ref{algo-BBC} the process must have received at least one
  valid message $\aux{r-1}{b\_val_r}$ and its proofs $proofs$.
  Now given that $est_i = b\_val_r$ (i.e. the same value as the coin),
  the value $prev_r$ computed on line~\ref{IV-05} of the $\isvalid$ predicate
  is the same when $\isvalid$ is called with input round $r-1$ or $r$.
  Thus calling predicate $\isvalid(r-1, est_i, proofs)$
  is equivalent to calling $\isvalid(r, est_i, proofs)$, and
  given that $\auxvalues_i$ contains $proofs$ (line~\ref{BBC-26}), $est_i$ will satisfy the $\isvalid$ predicate for round $r$.

  Now consider a non-faulty process whose estimate $est_i$ in round $r$ was set to $b\_val_r$ in round $r-1$.
  In this case, by line~\ref{BBC-12} of Figure~\ref{algo-BBC} the process must have received at least $n-t$
  valid $\aux{r-1}{\neg b\_val_r}$ messages.
  Then by lines~\ref{IV-05} and~\ref{IV-07} of Figure~\ref{algo-is-safe} these same messages satisfy the $\isvalid$
  predicate for round $r$ and $b\_val_r$.  
\end{proof}

Note that Lemma~\ref{lem:enough} only says that estimates broadcast by non-faulty processes satisfy the $\isvalid$ predicate,
but does not ensure that these values actually satisfy the validity definitions, the following lemmas will show this.

\begin{lemma}
  \label{lem:propose}
  In round $r > 0$ non-faulty processes will only sign and broadcast {\sc aux} messages containing binary values proposed by non-faulty processes.
\end{lemma}
\begin{proof}
  By Lemma~\ref{lem:enough}, a process will only broadcast messages that satisfy the $\isvalid$ predicate.
  By lines~\ref{IV-03},~\ref{IV-06},~\ref{IV-07} of the $\isvalid$ predicate, in any round $r > 0$
  a binary value will only satisfy the predicate if the process has received at least $t+1$ signed {\sc aux} messages
  from different processes.
  Given that there are at most $t$ faults and by induction, non-faulty processes will only broadcast values proposed by non-faulty processes.
\end{proof}

The idea of the next lemma is to show that if a binary is never valid in a round $r$, then it will never be valid
in any round after $r$.

\begin{lemma}
  \label{lem:valid}
  If a round $r_f > 0$, a binary value $b\_val$, and any set $proofs$ of signed {\sc aux} messages never satisfy
  the $\isvalid$ predicate, then then predicate will never be satisfied for $b\_val$ and any
  round $r_n > r_f$.
\end{lemma}
Given the construction of the $\isvalid$ predicate, for rounds following $r_f$ the predicate will output the same result as it
would for round $r_f$ until one round after the coin flip outputs value $\neg b\_val$
(i.e. until the value for $prev\_r$ is computed as a new value on line~\ref{IV-05}).

Now let $r_n$ be the first round after $r_f$ where the value of the coin is $b\_val$, (i.e.
$prev\_r$ is computed to be $r_n-1$).
Given that $b\_val$ is not valid in rounds $r_f,\ldots,r_n-1$, by lines~\ref{BBC-13}-\ref{BBC-15} no non-faulty process will
set $b\_val$ as its estimate and will not broadcast an {\sc aux} message containing $b\_val$ in these rounds.

As a result no process will receive more than $t$ signed {\sc aux} messages containing $b\_val$ in these rounds
and none of the lines of the $\isvalid$ predicate will be satisfied as they require at least $t+1$ messages.
By induction the same argument holds true for all following rounds.

\begin{proof}
  
\end{proof}

\begin{lemma}
  \label{lem:dec}
All non-faulty processes decide the same value.
\end{lemma}
\begin{proof}
  Assume a non-faulty process decides a binary value $b\_val$ in round $r_x$.
  By line~\ref{BBC-11} the process must have received $n-t$ signed $\aux{r_x}{b\_val}$ messages from different processes
  and the output of the coin for round $r_x$ must have been $b\_val$.
  Also by line~\ref{BBC-11} for this or a different non-faulty process to decide $\neg b\_val$, the process must
  receive $n-t$ signed $\aux{r_y}{\neg b\_val}$ messages from different processes in some round $r_y$.
  Furthermore by lines~\ref{BBC-21} and the \emph{c-global} property of the coin we have $r_y \neq r_x$.

  First assume $r_y > r_x$.
  By Lemma~\ref{lem:maj}, no process will receive $n-t$ signed $\aux{r_x}{\neg b\_val}$ messages from different processes and
  by Lemma~\label{lem:enough} and line~\ref{BBC-10} of Figure~\ref{algo-BBC},
  a non-faulty process will only sign and broadcast a value that satisfies the $\isvalid$ predicate.

  Given that there are less than $n-t$ signed $\aux{r_x}{\neg b\_val}$ messages from different processes,
  $\neg b\_val$ will never be valid in round $r_x$ (lines~\ref{IV-05}, \ref{IV-08} of the $\isvalid$ predicate)
  and by Lemma~\ref{lem:valid}, will not be valid in any following round.
  Thus, by line~\ref{BBC-20} $\neg b\_val$ will not be decided in any round after $r_x$
  (note that the case on line~\ref{IV-06} of the $\isvalid$ does not apply here because the value for $prev\_r$ computed for
  $\neg b\_val$ will always be at least $r_x$).

  Next assume $r_y < r_x$.
  If a process receives $n-t$ signed $\aux{r_y}{\neg b\_val}$ from different processes and decides ${\neg b\_val}$ in round $r_y$
  then using the same argument as above, no non-faulty process will receive $n-t$ signed $\aux{r_x}{b\_val}$ in any following round
  and will not decide $b\_val$.
  Thus by contradiction no process will decide ${\neg b\_val}$ in a round prior to $r_x$.
\end{proof}

\begin{lemma}
  \label{lem:nxtdec}
  Let $r_f$ be the smallest round in which a non-faulty process decides and the value of the coin in this round
  be $b\_val$. All non-faulty processes will decide in either
  round $r_f$ or the first round $r_n > r_f$ where the value the value of the coin in $r_n$ is $b\_val$.
\end{lemma}
\begin{proof}
  Given line~\ref{BBC-21} of Figure~\ref{algo-BBC}, a non-faulty process decides $b\_val$ in round $r_f$ after receiving
  $n-t$ signed $\aux{r_f}{b\_val}$ from different processes.
  By Lemma~\ref{lem:maj} no process will receive $n-t$ signed $\aux{r}{\neg v}$ messages from different processes,
  and by lines~\ref{BBC-05}-\ref{BBC-07} of the $\isvalid$ predicate, ${\neg b\_val}$ will never be valid in round
  $r_f$.
  Furthermore, given Lemma~\ref{lem:valid}, ${\neg b\_val}$ will not
  be valid in round any round after $r_f$.
  From this and by Lemma~\ref{lem:enough}, in all rounds after $r_f$ all non-faulty processes will broadcast
  messages containing the binary $b\_val$.
  Thus by line~\ref{BBC-20} of Figure~\ref{algo-BBC} in the first round $r_n$ after $r_f$ where the value of the coin
  is $b\_val$, all non-faulty processes will wait until they receive $n-t$
  signed $\aux{r_n}{b\_val}$ messages from different processes, and decide on line~\ref{BBC-12}.
\end{proof}

\begin{lemma}
  \label{lem:term}
  The algorithm terminates in expected $O(1)$ rounds.
\end{lemma}
\begin{proof}
  Given the \emph{c-threshold} property of the coin and that $t < n/3$, the value of the coin will not be revealed
  in a round $r > 0$ until at least $t+1$ non-faulty processes have participated in computing the output of the coin,
  i.e. $t+1$ non-faulty processes have reached line~\ref{BBC-17} of Figure~\ref{algo-BBC}.
  Consider the following two possible cases at the point where the $t+1$th non-faulty process reaches this line,
  just before the value of the coin for round $r$ is revealed.

  \begin{itemize}
  \item First assume that at least one of the $t+1$ non-faulty process has received $n-t$ valid $\aux{r}{b\_val}$
    messages for a single binary value $b\_val$ on line~\ref{BBC-12} and set its estimate to $b\_val$
    on line~\ref{BBC-14}.
    Now given \emph{c-random}, the output of the coin for round $r$ will be $b\_val$ with probability $1/2$ and
    the process will decide if it has not already done so on line~\ref{BBC-21}.

  \item Otherwise all $t+1$ non-faulty processes that have reached line~\ref{BBC-17} did not receive
    $n-t$ valid {\sc aux} for a single binary value, and as a result set their estimate to
    $\coinval$ on line~\ref{BBC-15}.
    Let the result of the coin for round $r$ be some binary value $b\_val$.
    In round $r+1$ these $t+1$ processes will have $b\_val$ as their estimate and broadcast the message $\aux{r+1}{b\_val}$.
    Given $t < n/3$, any set of $n-t$ valid {\sc aux} messages from round $r+1$ will contain
    at least one $\aux{r+1}{b\_val}$ message.
    Now given \emph{c-random}, the output for the coin in round $r+1$ will be the same binary value $b\_val$
    with probability $1/2$.
    In this case, all non-faulty processes will set their estimate to $b\_val$ in this round given that they received
    at least one $\aux{r+1}{b\_val}$ message (line~\ref{BBC-15} of Figure~\ref{algo-BBC}).
    Now given that any set of $n-t$ valid messages for round $r+1$ contains at least one $\aux{r+1}{b\_val}$ message,
    $\neg b\_val$ will never satisfy the $\isvalid$ predicate for round $r+2$ (lines~\ref{IV-05}-\ref{IV-07}) and
    given Lemma~\ref{lem:valid}, ${\neg b\_val}$ will not be valid in round any following round.
    From this and by Lemma~\ref{lem:enough}, in all rounds after $r+1$ all non-faulty processes will broadcast {\sc aux}
    messages containing the binary $b\_val$.
    Thus by line~\ref{BBC-20} of Figure~\ref{algo-BBC} in the first round $r_n$ after $r+1$ where the value of the coin
    is $b\_val$, all non-faulty processes will wait until they receive $n-t$
    signed $\aux{r_n}{b\_val}$ messages from different processes, and decide on line~\ref{BBC-12}
    if not already done.
  \end{itemize}

  Thus, in any round a non-faulty process will reach a state where termination is ensured
  with probability of at least $1/2$.
  Termination is then ensured with probailibty $1 - \prod_{r=1}^{\infty} 1/2 = 1$.
  Furthermore, the expected number of rounds to reach a state from which termination is ensured is
  $\sum_{r=1}^{\infty}r\frac{1}{2^n} = 2$, and by Lemma~\ref{lem:nxtdec} all processes will decide by
  the next round where the coin flip results in the same value, i.e. another expected $2$ rounds.

\end{proof}

\begin{theorem}
  The algorithm presented in Figure~\ref{algo-BBC} solves the Binary consensus problem in ${\BAMP}[t<n/3]$.
\end{theorem}
\begin{proof}
  First recall the definition of Binary Byzantine Consensus.
  \begin{itemize}
  \item BBC-Termination. Every non-faulty process eventually decides on a value.
  \item BBC-Agreement.   No two non-faulty processes decide on different values.
  \item BBC-Validity.  If all non-faulty processes propose the same value, no
    other value can be decided.
  \end{itemize}
  BBC-Termination is ensured by Lemma~\ref{lem:term}.
BBC-Agreement and BBC-Validity are ensured by Lemmas~\ref{lem:dec} and~\ref{lem:propose} respectively.
\end{proof}

\section{Implementation and experiments.}

\paragraph{Stopping and garbage collection.}
The algorithm shown in Figure~\ref{algo-BBC} continues to execute rounds forever.
To avoid this, if a non-faulty process decides in round $r$ it can simply broadcast a ``proof'' of decision, containing the $n-t$
messages that allowed it to decide and stop immediately.
Furthermore, the broadcast of this message may be delayed until the process receives a valid message from another process
from round $r+1$, ensuring that if all processes decide in round $r$ then no extra messages will be sent.
Note that in implementation, a process can not be immediately garbage collected as it needs
to ensure that its messages are reliably delivered
(reliable channels are often implemented through the use of re-transmissions when needed).
Fortunately, in a system that is executing multiple consensus instances, garbage collection of earlier
instances can be easily coordinated in later instances (this is not described here as it depends on
the requirements of the specific system).

\paragraph{Cryptographic signatures and validity proofs.}
Including proofs of validity with messages is necessary for the
correctness of the algorithm, but is not often needed in the expected case.
In fact in the presence of reliable channels the validity proofs are needed only in the
case of faults.
Given this, for efficiency an implementation may choose not to include proofs with messages by default
and instead have processes request proofs from the sender of the message if the recipient
cannot validate the message itself.
At worst, this slows down the execution of the algorithm as a non-faulty process may
have to wait to receive proofs from another non-faulty process.
Furthermore note that most proofs are a set of $n-t$ signatures of a single {\sc aux} message,
and given that the algorithm uses $n-t$ threshold signatures for the coin messages,
the same public keys can be used to sign proofs, reducing the proofs to a single signature in most cases.

\paragraph{Reducing the message steps}
In each round a non-faulty process broadcasts an {\sc aux} message,
waits until is receives $n-t$ valid {\sc aux} messages,
then broadcasts a {\sc coin} message, and waits
to receive $n-t$ valid {\sc coin} messages before continuing to the next round,
meaning each round includes the latency of at least $2$ message propagations.

For round $r > 0$ this can be reduced to the latency of a single message propagation by combining
the {\sc coin} message from round $r$ with the {\sc aux} message of round $r+1$
and broadcasting them together.
Notice that before the coin is broadcast in round $r$ on line~\ref{BBC-17}, the estimate that will be broadcast
for round $r+1$ has already been computed on lines~\ref{BBC-13}-\ref{BBC-15}, and the messages that will be used to generate
its proofs have already been received.
At this point only the value of the coin for round $r$ is unknown, so in the case that $est_i$
was set to $\coinval$ on line~\ref{BBC-15} of Figure~\ref{algo-BBC} the node will broadcast an
{\sc aux} message containing $\coinval$ instead of a binary value.
When nodes receive {\sc aux} message containing $\coinval$ they will simply wait until
they know the value of the coin for this round, then use these messages as if they
contained the same binary value as the coin.
Proofs for both binary values must be included with the message.
Note that given {\sc aux} messages can now hold $3$ different values, a proof of validity
for a round $r$ and binary value $b\_val$ may contain messages of the form
$\aux{r}{\coinval}$ and $\aux{r}{b\_val}$ and as a result the proofs may contain
$n-t$ signatures instead of a single threshold signature.

Notice that this modification obviously does not alter either BBC-Agreement or BBC-Validity
as the logic of the algorithm is unchanged.
Furthermore BBC-Termination remains valid as there still are $t+1$ non-faulty nodes
who have computed their estimate before the value of the coin is revealed
as needed by Lemma~\ref{lem:term}.

\subsection{Experiments}\label{sec:exp}
The algorithm has been implemented using the Go~\cite{GO} programming language.
Reliable channels are implemented through message re-transmission.
All messages contain the same predefined unique 32 byte string so that signatures
cannot be reused between different experiments.
All received messages are stored to disk in an append only log allowing processes
to recover quickly after a crash failure.
Threshold signatures use the
implementation of threshold BLS~\cite{BLS04} included in the Kyber library~\cite{kyber}.

Experiments were run on Google Could Compute using $10$, $20$, $40$, and $80$ n1-standard-1
instances (3.75 GiB of memory, 1 vCPU - a single hardware Hyper-thread, local SSD, 2 Gbps maximum egress bandwidth).
The instances were spread evenly across ten regions
in Asia, Australia, Europe, North America, and South America.

In each experiment nodes run 5 ``warm-up'' instances of binary consensus, followed by
50 additional binary consensus instances from which the results are calculated, with the graphs
showing the minimum, maximum and average values.
For each binary consensus instance nodes choose a random binary proposal.

Note that given the large number of random variables in the experiments we expect to see
a large amount of variance in the results. Ideally we would run consensus many more times
to get more stable results, but were unable to due to budget constraints.
Instead, in order to reduce the effect of randomness, the node's proposals and the outputs of the coin
flips are chosen using a seeded random generator that is reused for each experiment.
Note that the coin is still generated as described using threshold signatures, just the output is not used.

Furthermore given the low CPU power of the nodes and high computation cost of cryptography we
expect to see better performance on more capable machines, though again we were unable
to do this here due to budget constraints.

\begin{figure*}[ht!]
  \begin{subfigure}{.5\textwidth}
    \includegraphics[scale=0.60]{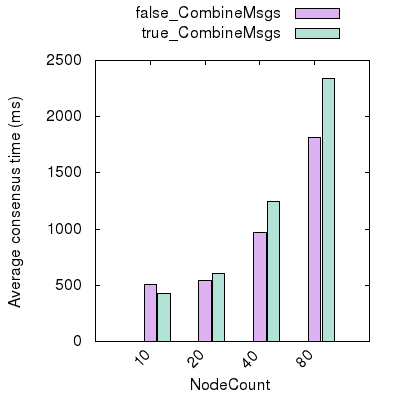}
    \caption{Latency}\label{fig:lat}
  \end{subfigure}
  \begin{subfigure}{.5\textwidth}
    \includegraphics[scale=0.60]{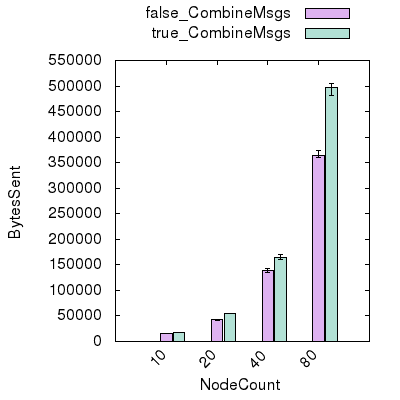}
    \caption{Bytes sent}\label{fig:bytes}
  \end{subfigure}
  \begin{subfigure}{.5\textwidth}
    \includegraphics[scale=0.6]{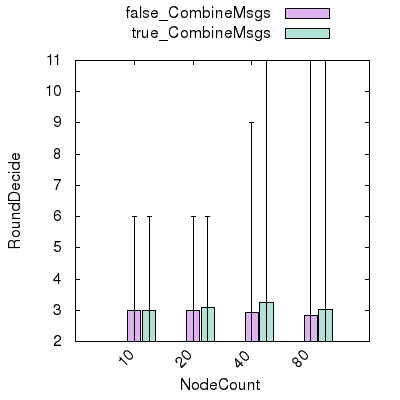}
    \caption{Decision round}\label{fig:round-dec}
  \end{subfigure}
  \begin{subfigure}{.5\textwidth}
    \includegraphics[scale=0.6]{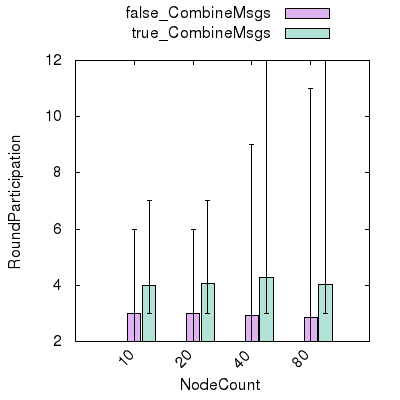}
    \caption{Participation round}\label{fig:round-term}
  \end{subfigure}  
  \caption{Experiment results with 10, 20, 40, and 80 single hardware thread nodes, where messages contain proofs of validity.}
  \label{fig:exp} 
\end{figure*}

Figure~\ref{fig:exp} shows the results of the experiment where \emph{false\_CombineMessages} are the results
of the standard algorithm and \emph{true\_CombineMessages} are the results with the optimization described previously combining
the {\sc coin} message with the {\sc aux} message of the following round.
Figure~\ref{fig:lat} shows the average latency of executing a single
consensus instance, Figure~\ref{fig:bytes} shows the average number
of bytes sent for a single consensus instance for all nodes,
Figure~\ref{fig:round-dec} shows the average, minimum, and maximum decision
round of the consensus instances, Figure~\ref{fig:round-term} shows the average, minimum, and maximum participation
(i.e. termination) round the consensus instances.

For \emph{false\_CombineMessages} and $10$ nodes we see average latencies around $500$ milliseconds
and \emph{true\_CombineMessages} being approximately $75$ milliseconds lower (Figure~\ref{fig:lat}).
As the number of nodes increases, the latency increases to over $2$ seconds, with \emph{true\_CombineMessages}
being slower than \emph{false\_CombineMessages}. This increase is largely created by the increase in computation needed to
validate signatures. To demonstrate this Figure~\ref{fig:exp-sleep} shows the latency results of the same
experiment, except where signature validations are replaced with sleeps of the estimated time to validate a signature,
where up to $4$ sleeps can be run concurrently (i.e. simulating a machine with $4$ processing cores).
Of course this is not completely realisitic as it does not simulate other operations that could slow down the execution
such as cache invlidations and garbage collection and is just for demonstration.
In this case the latency of \emph{true\_CombineMessages} remains lower than \emph{false\_CombineMessages}
by between $75$ to $100$ milliseconds and all averages stay below $600$ milliseconds.

\begin{figure*}[ht!]
\begin{center}
    \includegraphics[scale=0.60]{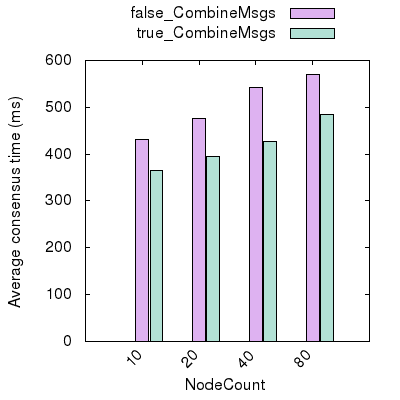}
  \end{center}
  \caption{Experiment results for 80 single hardware thread nodes, where up to 4 concurrent sleeps are
    performed instead signature validation to simulate a 4 core machine.}
  \label{fig:exp-sleep}
\end{figure*}

Concerning the number of rounds needed to decide, in all cases the average is approximately $3$ rounds (Figure~\ref{fig:round-dec})
with $2$ being the minimum. The maximum is $11$ rounds.
With \emph{false\_CombineMessages}, nodes terminate in the same round as they decide, while in the
case of \emph{true\_CombineMessages} nodes always participate in $1$ round following the round in which they decide (Figure~\ref{fig:round-term}).
This is simply because the coin message that results in the decision includes the message from the following round.
This in addition to the fact that the proofs of validity may be able to use threshold signatures as described previously
explains the increase in the number of bytes sent by \emph{true\_CombineMessages} (Figure~\ref{fig:bytes}),

\begin{figure*}[ht!]
  \begin{subfigure}{.5\textwidth}
    \includegraphics[scale=0.60]{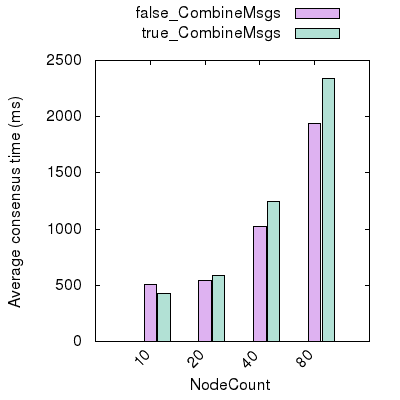}
    \caption{Latency}\label{fig:latnp}
  \end{subfigure}
  \begin{subfigure}{.5\textwidth}
    \includegraphics[scale=0.60]{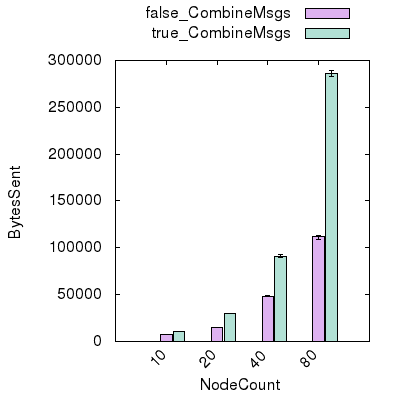}    
    \caption{Bytes sent}\label{fig:roundnp}
  \end{subfigure}
  \begin{subfigure}{.5\textwidth}
    \includegraphics[scale=0.6]{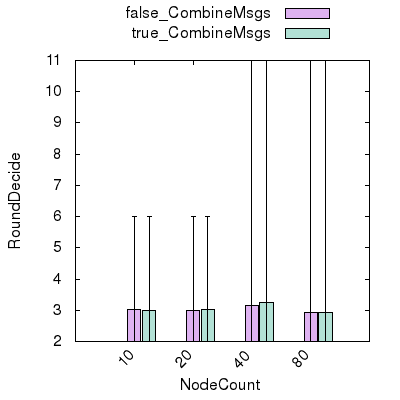}
    \caption{Decision round}\label{fig:round-decnp}
  \end{subfigure}
  \begin{subfigure}{.5\textwidth}
    \includegraphics[scale=0.6]{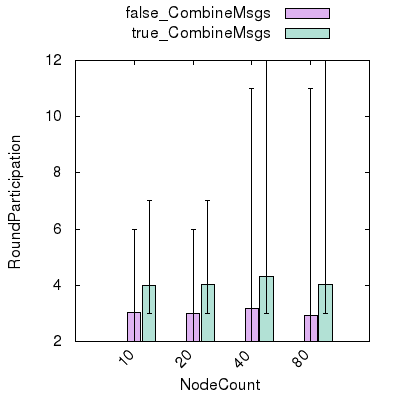}
    \caption{Participation round}\label{fig:roundtermnp}
  \end{subfigure}  
  \caption{Experiment results with 10, 20, 40, and 80 single hardware thread nodes, where messages do {\bf not} contain proofs of validity.}
  \label{fig:exp-noproofs} 
\end{figure*}

Figure~\ref{fig:exp-noproofs} shows the results of the same experiment as Figure~\ref{fig:exp}, except
here messages do not contain proofs of validity.
As mentioned previously, the proofs are only needed in the case of faults and can be sent on request
when needed by the receiver node.
Overall the results are fairly similar, with the main difference being
that the number of bytes sent is greatly reduced.
This is no surprise given that the consensus is over a binary value
and the main payload of the messages are the signatures.

\begin{figure*}[ht!]
  \begin{subfigure}{.5\textwidth}
    \includegraphics[scale=0.60]{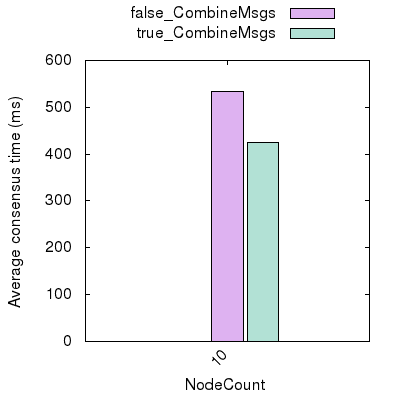}
    \caption{Latency}\label{fig:lat10}
  \end{subfigure}
  \begin{subfigure}{.5\textwidth}
    \includegraphics[scale=0.60]{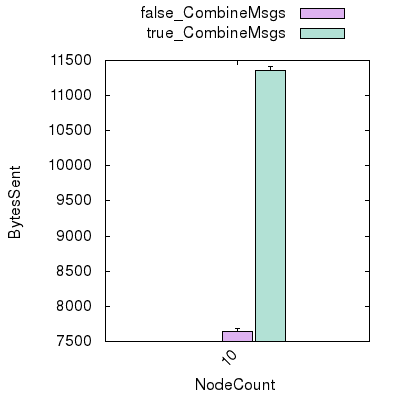}    
    \caption{Bytes sent}\label{fig:round10}
  \end{subfigure}
  \begin{subfigure}{.5\textwidth}
    \includegraphics[scale=0.6]{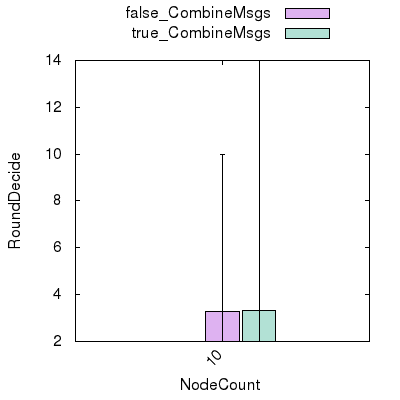}
    \caption{Decision round}\label{fig:round-dec10}
  \end{subfigure}
  \begin{subfigure}{.5\textwidth}
    \includegraphics[scale=0.6]{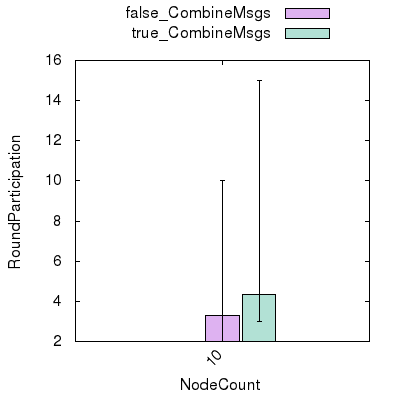}
    \caption{Participation round}\label{fig:round-term-10}
  \end{subfigure}  
  \caption{Experiment results for 1000 consensus instances with 10 nodes each with 4 hardware threads.}
  \label{fig:exp-10} 
\end{figure*}

Finally Figure~\ref{fig:exp-10} shows the results of an experiment with $10$ nodes
each with $4$ hardware threads, running $1000$ consensus instances and without including proofs of
validity.
Furthermore, the actual values from the coin flips are used.
The idea here is to perhaps have a more realistic experiment as more powerful nodes are used
and consensus is run many more times (only $10$ nodes are used due to budget constraints).
Here we see similar results as the previous experiments, except with somewhat
higher average latencies and higher maximum and average termination rounds,
which should be expected given the randomness of the experiments and the increased number of executions.

\section*{Acknowledgement}
Thanks to Daniel Collins for advice on cryptography.

\end{document}